\documentclass[a4paper,UKenglish,cleveref, autoref, thm-restate]{lipics-v2021}

\title{The Consistency of Probabilistic Databases with Independent Cells}

\usepackage{cite}
\usepackage{amssymb}
\usepackage{xcolor,colortbl}
\usepackage{listings}
\usepackage{makecell}
\usepackage{varwidth}
\usepackage{multirow}
\usepackage[bibliography=common]{apxproof}

\definecolor{Gray}{gray}{0.85}

\newcolumntype{a}{>{\columncolor{Gray}}c}
\newcolumntype{e}{>{\columncolor{white}}c}
\newcolumntype{d}{>{\columncolor{Gray}\centering} p{1.9in}}

\def\thickhline{\Xhline{3\arrayrulewidth}}

\def\qedexample{\hfill $\blacklozenge$}

\author{Amir Gilad}{Duke University, Durham, NC, USA} {agilad@cs.duke.edu}{}{}
\author{Aviram Imber}{Technion -- Israel Institute of Technology} {aviram.imber@cs.technion.ac.il}{}{}
\author{Benny Kimelfeld}{Technion -- Israel Institute of Technology} {bennyk@cs.technion.ac.il}{}{}

\authorrunning{A.~Gilad, A.~Imber, and B.~Kimelfeld} 

\Copyright{A.~Gilad, A.~Imber, and B.~Kimelfeld} 



\category{} 

\relatedversion{} 

\funding{The work of Amir Gilad was supported by NSF awards IIS-1703431, IIS-1552538, and IIS-2008107.
The work of Aviram Imber and 
Benny Kimelfeld was supported by the Israel Science
  Foundation (ISF), Grant 768/19, and the German Research Foundation (DFG) Project 412400621 (DIP program).
}

\nolinenumbers 

\EventEditors{John Q. Open and Joan R. Access}
\EventNoEds{2}
\EventLongTitle{42nd Conference on Very Important Topics (CVIT 2016)}
\EventShortTitle{CVIT 2016}
\EventAcronym{CVIT}
\EventYear{2016}
\EventDate{December 24--27, 2016}
\EventLocation{Little Whinging, United Kingdom}
\EventLogo{}
\SeriesVolume{42}
\ArticleNo{23}
\hideLIPIcs

\usepackage{mathtools}
\newcommand{\defeq}{\vcentcolon=}
\newcommand{\eqdef}{=\vcentcolon}

\def\true{\mathsf{true}}
\def\false{\mathsf{false}}

\def\prob{\mathrm{Pr}}

\def\set#1{\mathord{\{#1\}}}
\def\angs#1{\mathord{\langle#1\rangle}}

\newcommand{\eat}[1]{}

\def\e#1{\emph{#1}}

\def\ra{\rightarrow}
\def\la{\leftarrow}
\def\lra{\leftrightarrow}

\def\true{\mathbf{true}}
\def\false{\mathbf{false}}

\def\closure#1#2{#1^+_{#2}}

\def\vals{\mathbf{Val}}
\def\atts{\mathbf{Att}}
\def\tids{\mathsf{tids}}

\def\support{\mathsf{supp}}

\def\fpsharpp{\mathord{\mathrm{FP}^{\mathrm{\#P}}}}
\def\sharpp{\mathrm{\#P}}

\def\uca#1{\mathord{{?}#1}}

\def\tA{\uca A}
\def\tB{\uca B}
\def\tC{\uca C}
\def\uatts{\mathord{\uca\atts}}

\def\att#1{\textsf{#1}}
\def\uatt#1{\uca{\att{#1}}}

\def\val#1{\texttt{#1}}

\newtheoremrep{rlemma}[lemma]{Lemma}
\newtheoremrep{rtheorem}[theorem]{Theorem}
\newtheoremrep{rproposition}[proposition]{Proposition}
\newtheoremrep{rcorollary}[corollary]{Corollary}

\begin{document}

\maketitle

\begin{abstract}
A probabilistic database with attribute-level uncertainty consists of relations where cells of some attributes may hold probability distributions rather than  deterministic content. Such databases arise, implicitly or explicitly, in the context of noisy operations such as missing data imputation, where we automatically fill in missing values, column prediction, where we predict unknown attributes, and database cleaning (and repairing), where we replace the original values due to detected errors or violation of integrity constraints. We study the computational complexity of problems that regard the selection of cell values in the presence of integrity constraints. More precisely, we focus on functional dependencies and study three problems: (1) deciding whether the constraints can be satisfied by any choice of values, (2) finding a most probable such choice, and (3) calculating the probability of satisfying the constraints. The data complexity of these problems is determined by the combination of the set of functional dependencies and the collection of uncertain attributes. We give full classifications into tractable and intractable complexities for several classes of constraints, including a single dependency, matching constraints, and unary functional dependencies.
\end{abstract}

\section{Introduction}\label{sec:intro}

\eat{
\begin{figure}[!htb]
\centering
    \begin{subfigure}[b]{0.4\textwidth}
     \centering
    \begin{enumerate}
        \item $specialist?,time \rightarrow room$ 
        \item $room,time \rightarrow specialist?$
    \end{enumerate}
    \vspace{-3mm}
    \caption{Functional Dependencies.}\label{fig:fds-poly}
    \end{subfigure}
    \caption{Most probable imputation with FD (1) (NP-Hard): $t_1$[specialist] = Dr. Bart, $t_1$[specialist] = Dr. Lisa, and $t_1$[specialist] = Dr. Maggie with probability 14.4\%. Most probable imputation with FDs (1) and (2) (PTIME): $t_1$[specialist] = Dr. Bart, $t_1$[specialist] = Dr. Lisa, and $t_1$[specialist] = Dr. Bart  with probability 9.6\%.}\label{fig:ex-hard}
\end{figure}
}

Various database tasks amount to reasoning about relations where attribute values are uncertain. To name a few, systems for \e{data cleaning} may detect errors and suggest alternative fixes with different confidence scores~\cite{RekatsinasCIR17,DBLP:conf/icdt/SaIKRR19,DBLP:conf/sigmod/HeidariMIR19}, approaches to \e{data repair} may suggest alternative values due to the violation of integrity constraints (e.g., key constraints and more general functional dependencies)~\cite{AndritsosFM06,DBLP:journals/tods/Wijsen05}, and algorithms for \e{missing-data imputation} may suggest a probability distribution over possible completions of missing values~\cite{DBLP:journals/jmlr/BiessmannRSNSTL19,DBLP:conf/sigmod/MayfieldNP10}. 
Such uncertainty is captured as a probabilistic database in the
so called \e{attribute-level uncertainty}~\cite{DBLP:series/synthesis/2011Suciu} (as opposed to the commonly studied \e{tuple-level uncertainty}~\cite{DBLP:conf/vldb/DalviS04}). 

We refer to a relation of a probabilistic database in the attribute-level uncertainty as a Cell-Independent Relation (CIR). 
A CIR is a probabilistic database with a single relation, where the content of a cell is a distribution over possible values, and different cells are probabilistically independent. 
The CIR is the correspondent of a relation in the \e{Tuple-Independent Database} (TID)  under the \e{tuple-level uncertainty}, where the existence of each tuple is uncertain (while its content is certain), and different tuples are probabilistically independent~\cite{DBLP:series/synthesis/2011Suciu}. In contrast, the tuples of a CIR always exist, but their content is uncertain. 
For illustration, Figure~\ref{fig:specialist-cir} depicts a CIR with uncertain information about specialists attending rooms (e.g., since their attendance is determined by noisy sensors). 
Some attributes (here $\att{room}$ and $\att{business}$) are certain and have deterministic values. The uncertain attributes (e.g., $\uca{\att{specialist}}$) are marked by a question mark and their cells have several options for values. 
We later explain how this distinction has crucial impact on the complexity of CIRs.

\begin{figure}[t]
\centering
\begin{subfigure}{0.5\linewidth}
    \centering \footnotesize
        \begin{tabular}{ c | c | a | c | c | c |}
          \cline{2-4} \textit{tid} & \att{room} & \multicolumn{1}{c|}{$\uca{\att{specialist}}$} & \att{time}  \\
            \thickhline 1 & \val{41}  & \val{Bart}(0.5) $\mid$ \val{Lisa}(0.5) & \val{5 PM} \\
            \hline 2 & \val{163} & \val{Bart}(0.7) $\mid$ \val{Lisa}(0.3) & \val{5 PM} \\
            \hline 3 & \val{41}  & \val{Bart}(0.2) $\mid$ \val{Maggie}(0.8) & \val{5 PM}\\
            \hline
        \end{tabular}
    \caption{CIR $U_1$ with uncertain $\att{specialist}$.}\label{fig:specialist-cir}
    \end{subfigure}
    \begin{subfigure}{0.48\linewidth}\footnotesize
    \parbox{1\linewidth}{\hrule
    \begin{align*}
    F_1\defeq \{ & \uca{\att{specialist}}\;\att{time}\ra\att{room}\}\\
    F_2\defeq \{ & \uca{\att{specialist}}\;\att{time}\ra\att{room}\,,\\
     & \att{room}\;\att{time}\ra\uca{\att{specialist}}\}
    \end{align*}
    \hrule
    }
    \caption{Sets $F_1$ and $F_2$ of functional dependencies\label{fig:specialist-fds}}
    \end{subfigure}
    \vskip1.5em
\begin{subfigure}{0.8\linewidth}
    \centering \footnotesize
        \begin{tabular}{ c | c | c | c | c | c |}
          \cline{2-4} \textit{tid} & \att{room} & $\uca{\att{specialist}}$ & \att{time}  \\
            \thickhline 1 & \val{41}  & \val{Lisa}  & \val{5 PM} \\
            \hline 2 & \val{163} & \val{Bart} & \val{5 PM} \\
            \hline 3 & \val{41}  & \val{Maggie}  & \val{5 PM}\\
            \hline
        \end{tabular}\quad\quad
\begin{tabular}{ c | c | c | c | c | c |}
          \cline{2-4} \textit{tid} & \att{room} & $\uca{\att{specialist}}$ & \att{time}  \\
            \thickhline 1 & \val{41}  & \val{Bart}  & \val{5 PM} \\
            \hline 2 & \val{163} & \val{Lisa} & \val{5 PM} \\
            \hline 3 & \val{41}  & \val{Bart}  & \val{5 PM}\\
            \hline
        \end{tabular}
        \caption{Samples $r$ (left) and $r'$ (right) of $U_1$.}\label{fig:specialist-samples}
        \end{subfigure}
    
    \caption{Running example: CIR, FDs, and samples.}
\end{figure}



A natural scenario, studied by previous work for the TID model~\cite{GVSBUDA14,DBLP:journals/tods/LivshitsKR20}, considers a probabilistic database in the presence of a given set of integrity constraints, and specifically, Functional Dependencies (FDs).
Such a scenario gives rise to several interesting computational challenges, and we focus here on three basic ones.
In the problem of \e{possible consistency}, the goal is to test for the existence of a possible world (with a nonzero probability) that satisfies the FDs. 
The problem of the \e{most probable database} (``MPD''~\cite{GVSBUDA14}) is that of finding a possible world that satisfies the FDs and has the highest probability. 
In the problem of computing the \e{probability of consistency}, the goal is to calculate the above probability exactly (beyond just deciding whether it is nonzero), that is, the probability that (a random sample of) the given CIR satisfies the underlying FDs.
We investigate the computational complexity of these three problems for the CIR model. Our results provide classifications of tractability for different classes of FDs.
Importantly, we show that, for the studied classes, the complexity of these problems is determined by two factors: \e{(1)} the location of the uncertain attributes in the FDs (left or right side), \e{and (2)} the combination of the FDs in the given set of constraints.

The three problems relate to each other in the following manner.
To solve MPD, we need to be able to solve the possible consistency problem. 
The analysis of the probability of consistency sheds light on the possible consistency problem (is it fundamentally harder to compute the probability than to just determine whether it is nonzero?), but its importance goes beyond that. As we explain in Section~\ref{sec:problems}, computing this probability is useful to any type of constraints over CIRs, as the tractability of this probability implies that we can efficiently sample correctly from the conditional space of consistent samples.

Our study adopts the standard yardstick of \e{data complexity}~\cite{DBLP:conf/stoc/Vardi82}, where we fix the relational schema and the set of functional dependencies. The schema mentions not only what attributes are in the header of the relation, but also \e{which attribute is certain and which attribute is uncertain}. The complexity of the problems can be different for different combinations of schema and constraints, and we aim for a detailed understanding of which combinations are tractable and which are not.

\begin{example}\label{ex:intro1}
Consider again the CIR $U_1$ in Figure~\ref{fig:specialist-cir} along with the FD set $F_1$ of Figure~\ref{fig:specialist-fds}, consisting of a single FD. The FD says that at a specific time, a specialist can be found in only one location. 
Figure~\ref{fig:specialist-samples} shows a consistent sample $r$ of $U_1$ whose probability is $\prob(r)=0.5\cdot 0.7\cdot 0.8 = 0.28$. In particular, this probability is nonzero and, so, $U_1$ is possibly consistent. 
This sample has a maximal probability among the consistent samples; therefore, $r$ is a most probable database for $U_1$. 
Now, consider the FD set $F_2$ shown in Figure~\ref{fig:specialist-fds}, where the first FD is the one of $F_1$ and the second states that no two specialists should be in the same room at the same time. The sample $r$ in Figure~\ref{fig:specialist-samples} is no longer consistent, but $r'$ (in the same figure) is a consistent sample and also a most probable database.
In fact, $r'$ is the only consistent sample in this case, so 
the probability of consistency for $F_2$ turns out to be that of $r'$.
\qedexample
\end{example}

In contrast to the state of affairs in the attribute-level uncertainty, for \e{tuple-level uncertainty}  much more is known about MPD (i.e., finding the most likely instance of a probabilistic database conditioned on conformance to a set of FDs). In the case of tuple-independent databases, 
Gribkoff, Van den Broeck, and Suciu~\cite{GVSBUDA14} established dichotomy in the complexity of MPD for sets of unary FDs. This dichotomy has been generalized by Livshits, Kimelfeld and Roy~\cite{DBLP:journals/tods/LivshitsKR20} to a full classification over all sets of FDs, where they also established that the problem is equivalent to finding a \e{cardinality repair} of an inconsistent database. Carmeli et al.~\cite{DBLP:conf/icdt/CarmeliGKLT21} showed that two tractable cases, namely a \e{single FD} and a \e{matching constraint}, remain tractable even if the FDs are treated as \e{soft constraints} (where every violation incurs a cost). 
In this work, we aim to bring our understanding of attribute-level uncertainty closer to  tuple-level uncertainty.

\subparagraph*{Results.}
We establish classification results on several classes of functional dependencies: singleton FDs, matching constraints (i.e., FD sets of the form $\set{X\ra Y,Y\ra X}$), and arbitrary sets of unary FDs. Each classification consists of three internal classifications---one for each of the three problems we study (possible consistency, most probable database, and the probability of consistency). In every case, finding a most probable database is tractable whenever possible consistency is tractable. There are cases where the probability of consistency is intractable in contrast to the tractability of the most probable database, but we did not find any case where the other direction holds (and we will be surprised if such case exists). 
We also establish some general conclusions beyond these classes. For example, in Theorem~\ref{thm:all-uncertain-all-hard} (of Section~\ref{sec:single-matching}) we claim that if we make no assumption 
that some attributes are certain, then possible consistency is hard for \e{every nontrivial set of FDs}.

\begin{example}
Reconsider the CIR $U_1$ in Figure~\ref{fig:specialist-cir} along with the FD set $F_1$ of Figure~\ref{fig:specialist-fds}, consisting of a single FD.
Our classification shows that, in general, finding a solution to the possible consistency problem for such an FD, with uncertain attributes on the left side, is NP-complete. 
Now, reconsider the FD set $F_2$ shown in Figure~\ref{fig:specialist-fds}, where the first FD is the one of $F_1$. Thus, $F_1 \subset F_2$, however, interestingly,
our results show that for sets with the structure of $F_2$, finding an MPD (and, hence, also solving possible consistency) is in polynomial time. Intuitively, the additional FD in $F_2$ constrains the uncertain attribute on the left side of the first FD, making the problem tractable.
Finally, computing the probability of consistency for sets with the structure of $F_1$ and $F_2$ is $\sharpp$-hard (or more precisely $\fpsharpp$-complete).
\qedexample
\end{example}

\subparagraph*{Related work.}

A most probable database is the same as the ``Most Likely Intention'' (MLI) in the framework of Probabilistic Unclean Databases (PUD) of De Sa et al.~\cite{DBLP:conf/icdt/SaIKRR19}, in the special case where the \e{intention model} demands hard constraints and the \e{realization model} applies random changes to cells independently in what they refer to as \e{parfactor/update} PUD. They showed that finding an MLI of a parfactor/update PUD generalizes the problem of finding an \e{update repair} of an inconsistent database with a minimum number of value changes.
In turn, finding a minimal update repair has been studied in the literature and several complexity results are known for special cases of FDs, such as hardness (e.g., for the FD set $\set{A\ra B,B\ra C}$ due to Kolahi and Lakshmanan~\cite{DBLP:conf/icdt/KolahiL09}) and tractability (e.g., for lhs-chains such as $\set{A\ra B, AD\ra C}$ due to Livshits et al.~\cite{DBLP:journals/tods/LivshitsKR20}). There are, though, substantial differences between finding a most probable consistent sample of a CIR and finding an optimal update repair of an inconsistent database, at least in the variations where complexity results are known. First, they allow to select \e{any} value (from an infinite domain) for a cell, in contrast to the distributions of the CIR that can limit the space of allowed values; indeed, this plays a major role in past repair algorithms (e.g., Proposition 5.6 of~\cite{DBLP:journals/tods/LivshitsKR20} and Algorithm \textsc{FindVRepair} of~\cite{DBLP:conf/icdt/KolahiL09}). Second, they allow to change the value of \e{any} attribute and do not distinguish between uncertain attributes (where changes are allowed) and certain ones, as we do here; this is critical since, again, without such assumptions the problem is intractable for every nontrivial set of FDs (Theorem~\ref{thm:all-uncertain-all-hard}).

The problem of possible consistency does not have a nontrivial correspondence in the tuple-independent database model since, there, if there is any consistent sample then the subset that consists of all deterministic tuples (i.e., ones with probability one) is such a sample. The probability of consistency might be reminiscent of the problem of \e{repair counting} that was studied for subset repairs~\cite{DBLP:journals/jcss/LivshitsKW21,DBLP:conf/pods/CalauttiLPS22a}. Besides the fact that subset repairs are about tuple-level uncertainty (and no probabilities are involved), here we do not have any notion of \e{maximality} (while a repair is required to be a maximal consistent subset).

A CIR can be easily translated into a relation of a \e{block-independent-disjoint} (BID) probabilistic database~\cite{ReS07}. 
In a BID, every relation is partitioned into independent blocks of mutually exclusive tuples, each associated with a probability. This model has also been studied under the terms \e{dirty database}~\cite{AndritsosFM06} and x-tuples~\cite{MaslowskiW13,ChengCX08,MoCLCY13}. 
This translation implies that every upper bound for BIDs applies to CIRs, and the contrapositive: every hardness result that we establish (e.g., for the most probable database) extends immediately to BIDs; yet, it does not imply the other direction. Moreover, we are not aware of any positive results on inference over BIDs regarding integrity constraints. In addition, the translation from a CIR to a BID loses the information of which attributes are certain and which are uncertain, and as aforesaid, if we allow every attribute to be uncertain then the problem is hard for every nontrivial set of FDs (Theorem~\ref{thm:all-uncertain-all-hard}). 

\subparagraph*{Organization.} The remainder of the paper is organized as follows. We begin with preliminary definitions and notation (Section~\ref{sec:prelims}). We then define the CIR data model (Section~\ref{sec:cir}) and the  computational problems that we study (Section~\ref{sec:problems}). Next, we describe our analysis for the case of singleton and matching dependencies (Section~\ref{sec:single-matching}), and then the case of unary functional dependencies (Section~\ref{sec:unary}). Lastly, we give concluding remarks (Section~\ref{sec:conclusions}). Some of the proofs are omitted for lack of space, and they can be found in the Appendix. 


\section{Preliminaries}\label{sec:prelims}
We begin with preliminary definitions and notation.

\subparagraph*{Relations.}
We assume countably infinite sets $\vals$ of values and $\atts$ of attributes. A \e{relation schema} is a 
finite set $R=\set{A_1,\dots,A_k}$ of attributes.
An \e{$R$-tuple} is a function $t:R\ra\vals$
that maps each attribute $A\in R$ to a value that we denote by $t[A]$. A \e{relation} $r$ is associated with a relation schema, denoted
$\atts(r)$, a finite set of tuple identifiers, denoted $\tids(r)$, and a mapping from $\tids(r)$ to 
$\atts(r)$-tuples. (Note we allow for duplicate tuples, as we do not assume that the tuples of different identifiers are necessarily different.) 
We say that $r$ is a relation \e{over} the relation schema $\atts(r)$.
We denote by $r[i]$ the tuple that $r$ maps to the identifier $i$. Hence, $r[i][A]$ is the value that tuple $i$ has for the attribute $A$. 
As an example, Figure~\ref{fig:specialist-samples} (left) depicts a relation $r$ with
$\atts(r)=\set{\att{room},\uca{\att{specialist}}, \att{time}}$ 
(for now, the question mark in $\uca{\att{specialist}}$ should be ignored.) Here, $\tids(r)=\set{1,2,3}$ and $r[1][\att{room}]=\val{41}$. 

Suppose that $X$ is a set of attributes.
We denote by $\pi_X r$ the projection of $r$ onto $X$. More precisely, $\pi_X r$ is the relation $r'$ such that $\atts(r')=X$, $\tids(r')=\tids(r)$, and $r'[i][A]=r[i][A]$ for every $A\in\atts(r)\cap X$.
Observe that in our notation, $(\pi_X r)[i]$ is the projection of tuple $i$ to $X$.
As a shorthand notation, we write $r[i][X]$ instead
of $(\pi_X r)[i]$. For example, in Figure~\ref{fig:specialist-samples} we have $r[2][\att{room}\;\uca{\att{specialist}}]=(\val{163},\val{Bart})$.


\subparagraph*{Functional dependencies.}

A \e{functional dependency}, or FD for short, is an expression of the form $X\ra Y$ where $X$ and $Y$ are finite sets of attributes.
We say that $X\ra Y$ is \e{over} a relation schema $R$ if $R$ contains all mentioned attributes, that is, $X\cup Y\subseteq R$.
A relation $r$ satisfies the FD $X\ra Y$ over $\atts(r)$ if every two tuples that agree on $X$
also agree on $Y$. In our notation, we say that $r$ satisfies 
$X\ra Y$ if for every two tuple identifiers 
$i$ and $i'$ in $\tids(r)$ it holds that
$r[i][Y]=r[i'][Y]$ whenever $r[i][X]=r[i'][X]$.
A relation $r$ satisfies a set $F$ of FDs over $\atts(r)$, denoted 
$r \models F$, if $r$ satisfies every FD in $F$.

We use the standard convention that in instances of $X$ and $Y$ we may remove curly braces and commas. 
To compactly denote a set of FDs, we may also intuitively combine multiple FD expressions and  change the direction of the arrows. For example, the notation $A\lra B\la CD$ is a shorthand notation of 
$\set{A\ra B,B\ra A, CD\ra B}$. 

An FD $X \ra Y$ is \e{unary} if $X$ consists of a single attribute, and it is \e{trivial} if $Y\subseteq X$ (i.e., it is satisfied by every relation).  A
\e{matching constraint} (as termed in past work~\cite{DBLP:conf/icdt/CarmeliGKLT21}) is a constraint of the form $X\lra Y$, that is, the set $\set{X\ra Y,Y\ra X}$. 

The \e{closure} $F^+$ of a set $F$ of FDs is the set of all FDs that are implied by $F$ (or, equivalently, can be inferred by repeatedly applying the axioms of Armstrong). For example, $F^+$ includes all of the trivial FDs. The closure $\closure{X}{F}$ of a finite set $X$ of attributes is the set of all attributes $A$ such that $X\ra A$ is in
$F^+$. Two finite attribute sets $X$ and $Y$ are \e{equivalent} (w.r.t.~$F$) if $\closure{X}{F}=\closure{Y}{F}$, or in other words,
$X\ra Y$ and $Y\ra X$ are both in $F^+$. By a slight abuse of notation, we say that two attributes $A$ and $B$ are
\e{equivalent} if $\set{A}$ and $\set{B}$ are equivalent.
Finally, if $F$ is a set of FDs, then we denote by
$\atts(F)$ the set of all attributes that occur in either the left or right sides of rules in $F$.


\subparagraph*{Probability distributions.}
We restrict our study in this paper to finite probability spaces
$(\Omega,\pi)$ where $\Omega$ is a nonempty finite set of \e{samples} and 
$\pi:\Omega\ra [0,1]$ is a probability function satisfying $\sum_{o\in\Omega}\pi(o)=1$. 
The \e{support} of $\delta=(\Omega,\pi)$, denoted $\support(\delta)$, is the set of samples 
$o\in\Omega$ such that $\pi(o)>0$. 
We denote by $\prob_\delta(o)$ the probability $\pi(o)$.
We may write just $\prob(o)$ when $\delta$ is clear from the context.

\begin{figure}
    \centering \footnotesize
        \begin{tabular}{ c | c | a | a | c | c |}
          \cline{2-4} \textit{tid} & \att{business} & \multicolumn{1}{c}{$\uca{\att{spokesperson}}$} & \multicolumn{1}{c|}{$\uca{\att{location}}$} \\
            \thickhline $1$ & \val{S.~Propane}  & \val{Mangione}(0.6) $\mid$ \val{Strickland}(0.4) & \val{Arlen}(0.6) $\mid$ \val{McMaynerberry}(0.4) \\
            \hline $2$ & \val{Mega Lo Mart} & \val{Mangione}(0.45) $\mid$ \val{Thatherton}(0.55) & \val{Arlen}(0.5) $\mid$ \val{McMaynerberry}(0.5) \\
            \hline $3$ & \val{Mega Lo Mart} & \val{Mangione}(0.4) $\mid$ \val{Buckley}(0.6) & \val{Arlen}(0.55) $\mid$ \val{McMaynerberry}(0.45)\\
            \hline \multirow{2}{*}{$4$} & \multirow{2}{*}{\val{Get In Get Out}} & \multirow{2}{*}{\val{Peggy}(1.0)} &
            \multicolumn{1}{d|}{\val{Arlen}(0.35) $\mid$ \val{McMaynerberry}(0.3) $\mid$ \val{Dallas}(0.35)} \\
            \hline
        \end{tabular}
    \caption{CIR $U_2$ with spokesperson and location as the uncertain attributes.    }\label{fig:spokesperson-cir}
    \end{figure}

\section{Cell-Independent Relations}\label{sec:cir}
A \e{Cell-Independent Relation}, or \e{CIR} for short, is similar to an ordinary relation, except that in certain attributes the values may be probabilistic; that is, instead of an ordinary value, each of them contains a probability distribution over values. One could claim that the model should allow every attribute to have uncertain values. However, knowing which attributes are certain has a major impact on the complexity of operations over CIRs. Formally, a CIR $U$ is defined similarly to a relation, with the following differences:
\begin{itemize}
    \item The schema of $U$, namely $\atts(U)$, has \e{marked attributes} where uncertain values are allowed. We denote a marked attribute using a leading question mark, as in $\uca A$, and the set of marked  attributes  by $\uatts(U)$. (Note that  $\uatts(U)$ is a subset of $\atts(U)$.)
    \item For every $i\in\tids(U)$ and marked attribute $\uca A\in \uatts(U)$, the cell $U[i][\uca A]$ is a probability distribution over $\vals$.
\end{itemize}

By interpreting cells as probabilistically independent, 
a CIR $U$ represents a probability distribution over ordinary relations. Specifically, a sample of $U$ is a relation that is obtained from $U$ by sampling a value for each uncertain cell. More formally, a sample of $U$ is a relation $r$ such that $\atts(r)=\atts(U)$, $\tids(r)=\tids(U)$, and for every $i\in\tids(r)$ and unmarked attribute $A$ we have that $r[i][A]=U[i][A]$.

The probability $\prob_U(r)$ of a sample $r$ of $U$ is the product of the probabilities of the values chosen for $r$:
\[\prob_U(r)=\prod_{i\in\tids(U)}\prod_{\uca A\in \uatts(U)}
\prob_{U[i][\uca A]}(r[i][\uca A])
\]
Note that $\prob_{U[i][\uca A]}(r[i][\uca A])$ is the probability of the value $r[i][\uca A]$ (i.e., the value that tuple $i$ of $r$ has for the attribute $\uca A$) according to the distribution $U[i][\uca A]$ (i.e., the distribution that 
tuple $i$ of $U$ has for the attribute $\uca A$).

\begin{example}\label{example:cir}
Figures~\ref{fig:specialist-cir} and~\ref{fig:spokesperson-cir} depict examples $U_1$ and $U_2$, respectively, of CIRs. 
$U_1$ has been discussed in Example \ref{ex:intro1} and $U_2$ describes a CIR that stores businesses along with their spokespeople and headquarters locations. Some information in $U_2$ is noisy (e.g., since the rows are scraped from Web pages), and particularly the identity of the spokesperson and the business location. 
$U_1$ has a single uncertain attribute, namely $\uatt{specialist}$, and $U_2$ has two uncertain attributes, namely $\uatt{spokesperson}$ and $\uatt{location}$. In particular, we have:
$$\atts(U_1)=\set{\att{room},\uca{\att{specialist}}, \att{time}} 
\quad\quad \uatts(U_1)=\set{\uatt{specialist}}$$ 
Distributions over values are written straightforwardly in the examples. For example, the distribution $U_1[2][\uatt{specialist}]$ is the uniform distribution that consists of $\val{Bart}$ and $\val{Lisa}$, each with probability $0.5$. 

The relations $r$ and $r'$ of Figure~\ref{fig:specialist-samples} are samples of $U_1$. By the choices made in $r$, the probability
$\prob_U(r)$ is $0.5\cdot 0.7\cdot 0.8$. 
Note that the probability of $r$ is smaller than the probability of the sample where the specialists are $\val{Lisa}$, $\val{Bart}$ and $\val{Maggie}$, for instance, respectively.
\qedexample 
\end{example}

\subparagraph{Simplified notation.}
In the analyses that we conduct in later sections, we may simplify the notation when defining a CIR $U$. When $\atts(U)=\set{A_1,\dots,A_k}$, we may introduce a new tuple $t[i]$ with $t[i][A_\ell]=a_\ell$ simply as
$(a_1,\dots,a_k)$, assuming that the attributes are naturally ordered alphabetically by their symbols. For example, if 
$\atts(U)=\set{A,B,C}$, then $(a,b,c)$ corresponds to the tuple that maps $A$, $B$ and $C$ to $a$, $b$ and $c$, respectively. We can also use a distribution $\delta$ instead of a value $a_\ell$. In particular, we write
$b_1|\dots|b_t$ to denote a uniform distribution among the values $b_1,\dots,b_t$.
\begin{example}
Continuing Example~\ref{example:cir}, in the simplified notation the tuple $U_1[2]$ can be written as $(\val{163},\val{Bart}|\val{Lisa},\val{5 pm})$ since the attributes are ordered lexicographically and, again, the distribution happens to be uniform. \qedexample
\end{example}

\paragraph*{Consistency of CIRs}
Let $F$ be a set of FDs and let $U$ be a CIR, both over the same schema.
A \e{consistent sample} of  $U$ is a relation $r\in\support(U)$ such that $r\models F$. We say that $U$ is \e{possibly consistent} if at least one consistent sample exists. 
By \e{the probability of consistency}, we refer to the probability $\prob_{r\sim U}(r\models F)$ that a random sample of $U$ satisfies $F$. As a shorthand notation, we denote this probability by $\prob_U(F)$. Note that $U$ is possibly consistent if and only if $\prob_U(F)>0$. A consistent sample $r$ is a \e{most probable database} (using the terminology of Gribkoff, Van den Broeck and Suciu~\cite{GVSBUDA14}) if $\prob(r)\geq\prob(r')$ for every other consistent sample $r'$.

\begin{example}
Consider the CIR $U_1$ of Figure~\ref{fig:specialist-cir}. Let $F_1$ 
be that of Figure~\ref{fig:specialist-fds}, saying that at a specific time, a specialist can be found in only one location.
Figure~\ref{fig:specialist-samples} (left) shows a consistent sample $r$ of $U_1$. Then
$\prob(r)=0.5\cdot 0.7\cdot 0.8 = 0.28$. In particular, this probability is nonzero, hence $U_1$ is possibly consistent. The reader can verify that $r$ has a maximal probability among the consistent samples (and, in fact, among all samples); therefore, $r$ is a most probable database for $U_1$. To calculate the probability of consistency, we will take the complement of the probability of \e{inconsistency}. An inconsistent sample can be obtained in two ways:
\e{(1)} selecting $\val{Lisa}$ in both the first and second tuples,  \e{or (2)} selecting \val{Bart} in the second tuple and in at most one of the first and the third (which we can compute as the complement of the product of the probabilities of selecting the others). Therefore,
$$\prob_{U_1}(F_1) = 1-\big(0.3\cdot 0.5 + 0.7\cdot (1-0.5\cdot 0.8)\big) \,.$$

Now suppose that we use $F_2$ of Figure~\ref{fig:specialist-fds} saying that, in addition to $F_1$, a room can host only one specialist at a specific time.
In this case, $r$ is no longer a consistent sample since Room \val{41} hosts different specialists at \val{5 PM}, namely \val{Lisa} and \val{Maggie}. The reader can verify that the only consistent sample now is $r'$ of Figure~\ref{fig:specialist-samples}.
In particular, $U_1$ remains possible consistent, the sample $r'$ is the most probable database, and the probability of consistency is the probability of $r'$, namely $0.5\cdot 0.3\cdot 0.2$.
\qedexample
\end{example}
\section{Consistency Problems}\label{sec:problems}
We study three computational problems in the paper, as in the following definition. 

\begin{definition}\label{def:problems}
Fix a schema $R$ and a set $F$ of FDs over $R$. In each of the following problems, we are given as input a CIR $U$ over $R$:
\begin{enumerate}
    \item \underline{Possible-consistency}: determine whether $\prob_U(F)>0$.
    \item \underline{Most probable database}: find a consistent sample with a maximum probability.
    \item \underline{Probability of consistency}: calculate $\prob_U(F)$.
\end{enumerate}
\end{definition}


Observe that these problems include the basics of probabilistic inference: \e{maximum likelihood} computation and \e{marginal probability} calculation. An MPD can be viewed as an optimal completion of missing values, or an optimal correction of values suspected of being erroneous, assuming the independence of cells (as a \e{prior} distribution) and conditioned on satisfying the constraints (as a \e{posterior} distribution). A necessary condition for the tractability of the most probable database is possible consistency, where we decide whether at least one consistent sample exists. The problem of computing the probability of consistency can be thought of as a basic problem that sheds light on possible consistency. For example, if possible consistency is decidable in polynomial time in some case, is it because we can, generally, compute the probability of consistency or because there is something fundamentally easier with feasibility? We will see cases that feature both phenomena. 

A more technical reason to why we wish to be able to compute the probability of consistency is that it guarantees the ability to \e{sample} soundly from the conditional probability distribution (the posterior), that is, have an efficient randomized algorithm that produces a consistent sample $r$ with the probability $\prob_U(r\mid r\models F)$. The idea is quite simple and applies to every condition $F$ over databases, regardless of being FDs  (and was used in different settings, e.g.,~\cite{DBLP:journals/tods/CohenKS09}). For completeness of presentation, we give the details in the Appendix. 

\begin{toappendix}
\section{Sampling Consistent Relations}
As stated in Section~\ref{sec:problems}, we now explain how computing the probability of consistency can help to sample correctly (i.e., with the correct probability) from the conditional probability where the condition is the satisfaction of a set $F$ of constraints. 
 Suppose that $X_1,\dots,X_N$ are the random elements that represent the distributions of the uncertain cells of a CIR $U$. To produce a random sample $r$ with probability $\prob(r\mid F)$, we sample from the $X_i$ one by one, from $X_1$ to $X_N$. When we sample from $X_j$,  the probability of each value $a$ is adjusted to be 
$\prob(X_i=a\mid F \land X_1,\dots,X_{j-1}=a_1,\dots,a_{j-1})$
where $a_1,\dots,a_{j-1}$ are the values chosen already for $X_1,\dots,X_{j-1}$.\footnote{This is true because the probability of a sample $r$ of $U$ can be written as $\prob(X_1,\dots,X_N=a_1,\dots,a_N\mid F)$, which is equal to $\prod_{j=1}^N \prob(X_j=a_j\mid F\land X_1,\dots,X_{j-1}=a_1,\dots,a_{j-1})$.} 
We can compute these adjusted probabilities if we know how to compute the probability of consistency. Specifically, by application of the Bayes rule, the adjusted probability can be represented as:
$$
\frac
{\prob(X_1,\dots,X_{j}=a_1,\dots,a_{j-1},a)\cdot\prob(F\mid X_1,\dots,X_{j}=a_1,\dots,a_{j-1},a)}
{\prob(X_1,\dots,X_{j-1}=a_1,\dots,a_{j-1})\cdot\prob(F\mid X_1,\dots,X_{j-1}=a_1,\dots,a_{j-1})}
$$
Then, the probabilities $\prob(F\mid X_1,\dots,X_{\ell}=b_1,\dots,b_{\ell})$ are simply $\prob_{U'}(F)$ where $U'$ is obtained from $U$ by replacing each $X_j$ with the deterministic $b_j$. 
\end{toappendix}

As aforesaid, the second and third problems are at least as hard as the first one: finding a most probable database of $U$ requires knowing whether $U$ is possibly consistent, and calculating the exact probability is at least as hard as determining whether it is nonzero. There is no reason to believe a-priori that their complexities are comparable. Yet, our analysis will show that the third has the same or higher complexity in the situations that we study.

\subsection{Complexity Assumptions}
In our complexity analysis, we will restrict the discussion to uncertain cells that are finite distributions represented explicitly by giving a probability for each value in the support. Note that if all uncertain cells of $U$ have a finite distribution, then $U$ has a finite set of samples. Yet, its size can be exponential in the number of rows of $U$ (and also in the number of columns of $U$, though we will treat this number as fixed as we explain next), even if each cell distribution is binary (i.e., has only two nonzero options). Every probability is assumed to be a rational number that is represented using the numerator and the denominator.

We will focus on the \e{data complexity} of problems, which means that we will make the assumption that the schema $R$ of the CIR and the set $F$ of FDs are both fixed. Hence, every combination $(R,F)$ defines a separate computational problem, and different pairs $(R,F)$ can potentially have different complexities.

\subsection{Preliminary Observations}
In the following sections, we study the complexity of the three consistency problems that we defined in Definition~\ref{def:problems}. 
Before we move on to the actual results, let us state some obvious general observations.
\begin{itemize}
    \item Possible consistency is in NP, since we can verify a ``yes'' instance $U$ in polynomial time by verifying that a relation $r$ is a consistent sample. 
    \item If possible consistency is NP-complete for some schema $R$ and set $F$ of FDs,  then it is NP-hard to find a most probable database, and it is NP-hard to compute the probability of consistency.
    \item We will show that the probability of consistency can be \#P-hard, or more precisely $\fpsharpp$-complete.\footnote{Recall that $\fpsharpp$ is the class of functions that are  computable in polynomial time with an oracle to a problem in $\sharpp$ (e.g., counting the number of satisfying assignments of a propositional formula).  This class is considered intractable, and above the polynomial hierarchy~\cite{DBLP:journals/siamcomp/Toda91}.}
Membership in $\fpsharpp$ of the probability of consistency is based on our assumption that probabilities are represented as rational numbers, and it can be shown using standard techniques (e.g.,~\cite{DBLP:conf/pods/GradelGH98,DBLP:conf/icdt/AbiteboulCKNS10}) that we do not repeat here. 
\end{itemize}
We will take the above for granted and avoid repeating the statements throughout the paper.
\section{Singleton and Matching Constraints}\label{sec:single-matching}
In this section, we investigate the complexity of the three problems we study in two special cases: a singleton constraint $\set{X\ra Y}$  and a \e{matching constraint} $X\lra Y$ (as it has been termed in past work~\cite{DBLP:conf/icdt/CarmeliGKLT21}). We give full classifications of when such constraints are tractable and intractable for the three problems. We note that we leave open the classification of the entire class of FD sets, but we provide it for the general case of unary FDs in Section~\ref{sec:unary}. 

We begin with the case of a binary schema, where every set of FDs is equivalent to either a singleton or a matching constraint. 

\subsection{The Case of a Binary Schema}\label{sec:binary}
Throughout this section, we assume that the schema is $\set{A,B}$. The complexity of the different cases of FDs is shown in Table~\ref{table:complexity:binary}. To explain the entries of the table, let us begin with the tractable cases.

\begin{table}
\centering
\caption{Complexity of the consistency problems for a binary schema. ``Possibility'' refers to possible consistency, ``MPD'' refers to the most probable database problem, and ``Probability'' refers to the probability of consistency.
\label{table:complexity:binary}}
\begin{tabular}{|c|c|c|c||c|}\hline
\textbf{FDs} & \textbf{Possibility} & \textbf{MPD} & \textbf{Probability} & \textbf{Propositions}  \\  
\hline
$A\ra\uca B$ & PTime & PTime & PTime & \ref{prop:A-ra-tB-poly} \\
$\uca A\ra B$ & NP-complete & NP-hard & $\fpsharpp$-complete & \ref{prop:hardness-tildeA-ra-B-and-tildaA-lra-tildeB} \\
$A\lra\uca B$ & PTime & PTime & $\fpsharpp$-complete & \ref{prop:A_lra_tildeB_poly} (PTime), \ref{prop:hardness-A-lra-tildeB} ($\fpsharpp$-c.) \\
$\uca A\lra \uca B$ & NP-complete & NP-hard & $\fpsharpp$-complete & \ref{prop:hardness-tildeA-ra-B-and-tildaA-lra-tildeB} \\
\hline
\end{tabular}
\end{table}

\subsubsection{Algorithms}\label{sec:algorithms-binary}
In this section, we show algorithms for $A\ra\uca B$ and for $ A\lra\uca B$.

For $A\ra\uca B$, we need to determine a value $b$ for each value $a$ of the attribute $A$. The idea is that we do so independently for each $a$. Let $V_A$ be the active domain of the attribute $A$ of $U$, and
$V_B$ be the set of all values in the supports of the distributions of $B$. Formally:
\begin{equation*}\label{eq:VA-VB}
V_A\defeq  \set{U[i][A]\mid i\in\tids(U)}\quad\quad
V_B\defeq  \bigcup\set{\support(U[i][B])\mid i\in\tids(U)}
\end{equation*}
A consistent sample $r$ selects a value $b_a\in V_B$ for each $a\in V_A$, and then
$\prob_U(r)=\prod_{a\in V_A} p(a,b_a)$
where $p(a,b)$ is given by:
\[p(a,b)\defeq \prod_{i :  U[i][A]=a}\prob_{U[i][B]}(b)\]
Therefore, to find a most probable database, we consider each $a\in V_A$ independently, and find a $b\in V_B$ that maximizes $p(a,b)$. This $b$ will be used for the tuples with the value $a$ in $A$. 
In addition,
we have the following formula that gives us immediately  a polynomial-time algorithm (via a direct computation) for the probability of consistency:
\[\prob_U(\set{A\ra B}) = \prod_{a\in V_A} \sum_{b\in V_B} p(a,b)\]
Where $\sum_{b\in V_B} p(a,b)$ is the probability that the tuples with the value $a$ for $A$ agree on their $B$ attribute. In summary, we have established the following.
\begin{proposition}\label{prop:A-ra-tB-poly}
All three problems in Definition~\ref{def:problems} are solvable in polynomial time for $A\ra\uca B$.
\end{proposition}

Next, we discuss $A\lra\uca B$.
Let $U$ be a CIR. A consistent sample $r$ of $U$ entails the matching of each $A$ value $a$ to each $B$ value $b$, so that no two $a$ values occur with the same $b$, and no two $b$s occur with the same $a$. Therefore, we can solve this problem using an algorithm for minimum-cost perfect matching, as follows. Let $V_A$, $V_B$ and $p(a,b)$ be as defined earlier in the section for $ A\ra\uca B$. 
We construct a complete bipartite graph $G$ as follows.
\begin{itemize}
    \item The left-side vertex set is $V_A$ and the right-side vertex set is $V_B$.
    \item The cost of every edge $(a,b)$ is $(-\log p(a,b))$; we use this weight as as our goal is to translate a maximum product into a minimum sum.\footnote{We assume that the computational model for finding a minimum-cost perfect matching can handle the representation of logarithms, including $\log 0=-\infty$. As an alternative, we could use directly an algorithm for maximizing the product of the edges in the perfect matching~\cite{Tseng:1993:FCM}.}
\end{itemize}
Note that $|V_A|$ and $|V_B|$ are not necessarily of the same cardinally. If $|V_A|>|V_B|$, then $U$ has no consistent sample at all. If $|V_A|<|V_B|$, then we add to the left side of the graph dummy vertices $a'$ that are connected to all $V_B$ vertices using the same cost, say $1$. 
With this adjustment, we can now find a most probable database by finding a minimum-cost perfect matching in $G$. In summary, we have established the following.
\begin{proposition}\label{prop:A_lra_tildeB_poly}
For $A\lra\uca B$, a most probable database can be found in polynomial time.
\end{proposition}

It turns out that the third problem, the probability of consistency, is intractable. We show it in the next section.

\subsubsection{Hardness}
We now discuss the hardness results of Table~\ref{table:complexity:binary}. 
We begin with $A\lra \uca B$. Recall that possible consistency and the most probable database are solvable in polynomial time (Proposition~\ref{prop:A_lra_tildeB_poly}). The probability of consistency, however, is hard. 
\begin{rpropositionrep}\label{prop:hardness-A-lra-tildeB}
For $A\lra \uca B$, it is $\fpsharpp$-complete to compute the probability of consistency. 
\end{rpropositionrep}

\begin{proof}
We show a reduction from the problem of counting the perfect matchings of a bipartite graph (which is the same as calculating the permanent of a 0/1-matrix). This problem is known to be $\sharpp$-complete~\cite{DBLP:journals/tcs/Valiant79}. We are given a bipartite graph $G=(V_L,V_R,E)$ such that $|V_L|=|V_R|$ and the goal is to compute the number of perfect matchings that $G$ has. 
We construct a CIR $U$ as follows. For each vertex $v\in V_L$ we collect the set $N_v\subseteq V_R$ of neighbors of $v$. Let $N_v=\set{u_1,\dots,u_\ell}$. We add to $U$ the tuple $(v,u_1|\dots| u_\ell)$. 

Observe that every consistent sample induces a perfect matching (due to $A\lra \uca B$), and vice versa. Hence, the number of consistent samples of $U$ is the same as the number of perfect matchings of $G$. Since we used only uniform probabilities, every sample of $U$ has the same probability, namely
$1/(\prod_{v\in V_L}|N_v|)$. Therefore, the number of perfect matchings is 
$\prob_U(A\lra\uca B)\cdot\prod_{v\in V_L}|N_v|$.
\end{proof}

The proof of Proposition~\ref{prop:hardness-A-lra-tildeB} is by a reduction from counting the perfect matchings of a bipartite graph, which is known to be $\sharpp$-complete~\cite{DBLP:journals/tcs/Valiant79}. The next proposition addresses the case of $\uca A\ra B$ and the case of $\uca A\lra \uca B$.

\begin{toappendix}
\begin{lemma}\label{lemma:hardness-tildeA-ra-B}
For $\uca{A}\ra B$:
\begin{enumerate}
    \item Possible consistency is NP-complete.
    \item It is $\fpsharpp$-complete to compute the probability of consistency. 
\end{enumerate}
\end{lemma}
\begin{proof}
We prove each part separately.

\subparagraph*{Part 1.}
 We show a reduction from \e{non-mixed satisfiability} (NM-SAT), where each clause contains either only positive literals (``positive clause") or only negative literals (``negative clause'').  This problem is known to be NP-complete~\cite{DBLP:conf/approx/Guruswami00}.

We are given a formula $c_1\land\dots\land c_m$ over
$x_1,\dots,x_n$. We construct an uncertain table as follows. For each positive $c_i=y_1\lor\dots\lor y_\ell$ we have in the table the tuple
$$(y_1|\dots|y_\ell,\true)\,,$$ 
that is, 
a tuple with a distinct identifier $i$ such that
$U[i][A]$ is a uniform distribution over 
$\set{y_1,\dots,y_\ell}$ and $U[i][B]$ is the value 
$\true$. 
Similarly, for each negative clause $c_i=\neg y_1\lor\dots\lor \neg y_\ell$ we have in the table the tuple
$$(y_1|\dots|y_\ell,\false)\,.$$
Hence, for each positive clause we need to select one satisfying variable, for each negative clause we need to select one satisfying variable, and we cannot select the same variable to satisfy both a positive and a negative clause. The correctness of the reduction is fairly obvious.

\subparagraph*{Part 2.}
To prove Part~2, we use a reduction from counting the perfect matchings, similarly to the proof of Proposition~\ref{prop:hardness-A-lra-tildeB}, except that now we reverse the order of the attributes:
Instead of adding the tuple $(v,u_1|\dots| u_\ell)$, we add the tuple
$(u_1|\dots| u_\ell,v)$. The reader can easily verify that each consistent sample again encodes a unique perfect matching, and vice versa.
 \end{proof}

\begin{rlemmarep}\label{lemma:hardness-tildeA-lra-tildeB}
For $\uca{A}\lra \uca B$:
\begin{enumerate}
    \item Possible consistency is NP-hard.
    \item It is $\fpsharpp$-complete to compute the probability of consistency. 
\end{enumerate}
\end{rlemmarep}
\begin{proof}
We prove each part separately.
\subparagraph*{Part 1.}
We need to show the NP-hardness of possible consistency.  We show a reduction from standard SAT, where we are given a formula $\varphi = c_1\land\dots\land c_m$ over $x_1,\dots,x_n$, and we construct a CIR $U$ over $\set{A,B}$ as follows. For each clause $c=d_1\lor\dots\lor d_\ell$ we add to $U$ the tuple $$(c,\angs{c,d_1}|\dots|\angs{c,d_\ell})\,.$$
Note that the values of $U$ are clauses $c$ and pairs $\angs{c,d}$ where $d$ is a literal. 
In addition to these tuples, we collect every two pairs $\angs{c,d}$
and $\angs{c',d'}$ such that $d$ and $d'$ are in conflict, that is,  if $d=x$ then $d'=\neg x$ and if $d=\neg x$ then $d'=x$. For each such pair, we add to $U$ the tuple
$$(\angs{c,d}|\angs{c',d'},\angs{c,d}|\angs{c',d'})\,.$$
This completes the reduction. Next, we prove the correctness of the reduction, that is, $\varphi$ is satisfiable if and only if $U$ is possibly consistent. 

For the ``only if'' direction, suppose that $\tau$ is a satisfying truth assignment for $\varphi$. We construct a consistent sample $r$ as follows. For every tuples of the form $(c,\angs{c,d_1}|\dots|\angs{c,d_\ell})$, we choose for $B$ a value $\angs{c,d_i}$ such that $\tau(d_i)=\true$. In the case of tuples of the form $(\angs{c,d}|\angs{c',d'},\angs{c,d}|\angs{c',d'})$, we choose the pair $\angs{c',d'}$ such that $\tau(d') = \false$ for both attributes $A$ and $B$. 
We need to show that $r$ satisfies $\uca A\lra \uca B$. It is easy to see why the left attribute determines the right attribute, and so, $\uca{A}\ra \uca B$ holds. Regarding $\uca{B}\ra \uca A$, we need to see verify that we do not have any conflicting tuples $(c,\angs{c,d})$ and $(\angs{c',d'},\angs{c',d'})$ where $c=c'$ and $d=d'$. This is due to the fact that $\tau(d)=\true$ and $\tau(d')=\false$. 

For the ``if'' direction, suppose that $r$ is a consistent sample. 
We define a satisfying truth assignment $\tau$ as follows. Suppose that $r$ contains
$(c, \angs{c,d})$. Then $r$ necessarily contains $(\angs{c',d'}, \angs{c',d'})$ for every $c'$ that contains the negation $d'$ of $d$. Therefore, $r$ does not contain any $(c', \angs{c',d'})$ where $d'$ contradicts $d$. So, we choose $\tau$ such that $\tau(d)=\true$. If needed, we complete $\tau$ to the remaining variables arbitrarily. From the construction of $\tau$ it holds that every clause $c$ is satisfied. This completes the proof of Part~1.

\subparagraph*{Part 2.}
Note that this part is immediate from Proposition~\ref{prop:hardness-A-lra-tildeB}, since every instance of 
$A\lra \uca B$ can be viewed as an instance of $\uca A\lra \uca B$ where all $A$ values are known.
\end{proof}
\end{toappendix}

\begin{rpropositionrep}\label{prop:hardness-tildeA-ra-B-and-tildaA-lra-tildeB}
For each of $\uca{A}\ra B$ and $\uca{A}\lra \uca B$:
\begin{enumerate}
    \item Possible consistency is NP-complete.
    \item The probability of consistency is $\fpsharpp$-complete.
\end{enumerate}
\end{rpropositionrep}
\begin{proof}
We proved each constraint separately in Lemma~\ref{lemma:hardness-tildeA-ra-B}
and Lemma~\ref{lemma:hardness-tildeA-lra-tildeB}, respectively. 
\end{proof}

 For possible consistency, we show reductions from variations of SAT, where we translate satisfying assignments of a formula to consistent samples of a CIR. For the probability of consistency we use a similar idea to Proposition~\ref{prop:hardness-A-lra-tildeB}. The proofs of Proposition~\ref{prop:hardness-A-lra-tildeB} and Proposition~\ref{prop:hardness-tildeA-ra-B-and-tildaA-lra-tildeB} can be found in the Appendix. 
 
 We have now completed all results of Table~\ref{table:complexity:binary}. We will use these results for the extension to singleton, matching, and unary constraints. 

\subsection{Singleton and Matching Constraints}
We generalize the results for the binary case to the more general case where the FD set is either a singleton or a matching constraint. 

\begin{rtheoremrep}\label{thm:singleton-matching-dichotomy}
Let $X$ and $Y$ be sets of attributes such that $X\not\subseteq Y$ and $Y\not\subseteq X$, and at least one attribute in $X\cup Y$ is uncertain. 
\begin{enumerate}
    \item In the case of $X\ra Y$, if $X$ consists of only certain attributes, then all three problems are solvable in polynomial time. Otherwise, possible consistency is NP-complete and the probability of consistency is $\fpsharpp$-complete.
    \item In the case of $X\lra Y$, if either $X$ or $Y$ consists of only certain attributes, then a most probable database can be found in polynomial time; otherwise, possible consistency is NP-hard. In any case,
    the probability of consistency is $\fpsharpp$-complete.
\end{enumerate}
\end{rtheoremrep}
\begin{proofsketch}
For the first part, the hardness side is due to a straightforward reduction from $\uca A\ra B$, where hardness is stated in Proposition~\ref{prop:hardness-tildeA-ra-B-and-tildaA-lra-tildeB}, and for the tractability side, we show a reduction to the case of $A\ra\tB$, which is tractable due to Proposition~\ref{prop:A-ra-tB-poly}. For the second part, the tractability side is via a reduction to the case of $A\lra\tB$, which is tractable due to Proposition~\ref{prop:A_lra_tildeB_poly}. The hardness of possible consistency relies on the cases of $\tA\ra B$ and $\tA\lra\tB$ from Proposition~\ref{prop:hardness-tildeA-ra-B-and-tildaA-lra-tildeB}, and the hardness of probability of consistency relies on the case $A\lra\tB$ from Proposition~\ref{prop:hardness-A-lra-tildeB}. The full proof can be found in the Appendix.
\end{proofsketch}

\begin{proof}
We prove each part separately. 
\subparagraph*{Part~1.} The hardness side is due to straightforward reduction from $\uca A\ra B$ where hardness is stated in Lemma~\ref{lemma:hardness-tildeA-ra-B}. For the tractability side, we show a reduction to the case of $A\ra\tB$, which is tractable due to Proposition~\ref{prop:A-ra-tB-poly}. Without loss of generality, suppose that all attributes of $Y$ are uncertain (and we assume that all attributes of $X$ are certain). Given an input $U$ for $X\ra Y$, we construct an instance $U_0$ for  $A\ra\tB$ with $\tids(U_0)=\tids(U)$ by converting each tuple of $U$ into a tuple of $U_0$. The construction is simple: for each identifier $i\in\tids(U)$, the tuple
$U_0[i]$ is obtained as follows:
\begin{itemize}
    \item $U_0[i](A)=\pi_X U[i]$ (i.e., the  concatenation of the values in the attributes of $X$);
    \item $U_0[i](\uca{B})$ is the distribution over all possible tuples that can be generated from
    $\pi_Y U[i]$, each with its probability (i.e., the product of the values of the different attributes of $Y$).
\end{itemize}
The correctness of the reduction is straightforward. 

\subparagraph*{Part~2.} 
The tractability side is via a reduction to the case of $A\lra\tB$, which is tractable due to Proposition~\ref{prop:A_lra_tildeB_poly}. The reduction is the same as the one to $A\ra\tB$ that we showed in Part~1.

For the hardness side, let us begin with possible consistency. We consider two cases. If $X$ and $Y$ share an uncertain attribute $\tA$, then there are easy reductions from the case of $\tA\lra\tA B$ (by fixing all other attributes to a single constant), and $\tA\lra\tA B$ is equivalent to $\tA\ra B$ where hardness is due to Lemma~\ref{lemma:hardness-tildeA-ra-B}. If $X$ and $Y$ do not share an uncertain attribute, then $X$ and $Y$ contain distinct uncertain attributes $\tA$ and $\tB$, and then we have easy reductions (again fixing all other attributes different from $\tA$ and $\tB$ to a single constant) from
the case of $\tA\lra\tB$, where hardness is due to Lemma~\ref{lemma:hardness-tildeA-lra-tildeB}.

The above arguments also imply $\fpsharpp$-completeness of the probability of consistency, except for the case where all attributes of $X$ are certain. In the latter case, we apply our easy reduction from $A\lra\tB$ and use Proposition~\ref{prop:hardness-A-lra-tildeB}.
\end{proof}

\begin{example}
Consider again the CIR $U_1$ of Figure~\ref{fig:specialist-cir}, and the following two constraints: $F_1\defeq\set{\uca{\att{specialist}}\;\att{time}\ra\att{room}}$ and $F_2\defeq F_1\cup\set{\att{room}\;\att{time}\ra\uca{\att{specialist}}}$.
For $F_1$, all three problems are hard, since the left hand side of the FD contains the uncertain attribute $\uca{\att{specialist}}$. For $F_2$, a most probable database can be found in polynomial time, since $F_2$ is equivalent to
$\att{room}\;\att{time}\lra \uca{\att{specialist}}\;\att{time}$, where one side (the left side) consists of only certain attributes.
However, the probability of consistency remains $\fpsharpp$-hard.
\qedexample
\end{example}

Note that in Theorem~\ref{thm:singleton-matching-dichotomy}, the assumption that $X\not\subseteq Y$ and $Y\not\subseteq X$  does not lose generality, for the following reason. If $X\subseteq Y$, then the FD $X\ra Y$ is equivalent to
$X\ra Y\setminus X$, the FD $Y\ra X$ is trivial, and the matching constraint
$X\lra Y$ is equivalent to the singleton $\set{X\ra Y}$ (which is covered in Part~1).


From Theorem~\ref{thm:singleton-matching-dichotomy} we can conclude that when all attributes are uncertain, possible consistency is hard, unless the FDs are all trivial (and then all three problems are clearly solvable in polynomial time); this is under the reasonable (and necessary) assumption that $F$ has no \e{consensus FDs}, that is, the left hand side of every FD is nonempty~\cite{DBLP:journals/tods/LivshitsKR20}. We later discuss this assumption.
This emphasizes the importance of having a data model that distinguishes between certain and uncertain attributes.

\begin{rtheoremrep}\label{thm:all-uncertain-all-hard}
Let $F$ be a nontrivial set of FDs over a relation schema $R$, none being a consensus FD. Then possible consistency is NP-complete.
\end{rtheoremrep}
\begin{proof}
Let $X\ra Z$ be a nontrivial FD in $F$ such that $X$ is minimal (with respect to containment) among the left hand sides of FDs in $F$. Let $Y=X^+_F\setminus X$.
Note that $X$ is nonempty due to the assumption of the theorem. Also note that $Y$ is nonempty since the closure of $X$ contains every attribute in $Z$. We consider two cases:
\begin{enumerate}
    \item $F\models Y\ra \tA$ for some $\tA\in X$.
    \item $Y_F^+=Y$.
\end{enumerate}
Note that these are the only possible two cases, since $X\cup Y$ is closed under $F$. 

We first select an attribute $\tA\in X$ as follows. 
In the first case, we use the $\tA$ of the case; in the second case, we choose an arbitrary $\tA\in X$ (which exists due to the assumption that $X$ is nonempty). In the first case we apply a reduction from $\tA\lra\tB$, and in the second we apply a reduction from $\tA\ra\tB$. 
Note that the two problems are hard for both constraints, according to Theorem~\ref{thm:singleton-matching-dichotomy}. The reduction is detailed next.

Let $U$ be an input CIR over $\set{\tA,\tB}$. We construct a CIR $U'$ over $R$
with $\tids(U')=\tids(U)$ by transforming every tuple $U[i]$ into a tuple $U'[i]$, as follows. Let $c\in\vals$ be an arbitrary constant value. Let $\delta$ be a uniform distribution over all pairs $\angs{a,b}$ where $a\in\vals$ and $b\in\vals$ are values in the (ranges of) distributions in $U[\tA]$ and $U[\tB]$, respectively. For every attribute $\uca{C}$ of $R$,
\[
U'[\uca{C}] \eqdef
\begin{cases}
U[i][\tA]  & \mbox{if $\uca{C}=\tA$;} \\
U[i][\tB]  & \mbox{if $\uca{C}\in Y$;} \\
c & \mbox{if $\uca{C}\in X\setminus\set{\tA}$;} \\
\delta & \mbox{otherwise.}
\end{cases}
\]
The correctness of the reduction is due to the following.
First, if $r'$ is a consistent sample of $U'$, then $\pi_{\set{\uca{A},\uca{B}}}r'$ is a consistent sample of $U$ due to the construction of $U'$. 

Second, every consistent sample $r$ of $U$ can be extended into a consistent sample $r'$ of $U'$ by choosing the value $\angs{r[i][\tA],r[i][\tB]}$ for each distribution $\delta$ in a cell of the tuple $i$. 
To see that $r'\models F$, consider a nontrivial FD $X'\ra Y'$ and let $i$ and $j$ be two tuple identifiers in $\tids(r')$. Suppose that $r'[i]$ and $r'[j]$ agree on the tuples of $X'$. Note that $X'$ cannot be a strict subset of $X$ due to the minimality of $X$. If $X'$ includes at least one attribute outside of $Y$, then 
$r'[i][X]$ and $r'[j][X]$ include an occurrence of $r[i][\tA]$ and, hence, they must be equal (due to the construction of $r'$ and the consistency of $r$). 
If $X'\subseteq Y$ and we are in the first case, then $r'[i][X]$ and $r'[j][X]$ include an occurrence of $r[i][\tB]$ and, hence, they must be equal. If $X'\subseteq Y$ and we are in the second case, then $Y'\subseteq Y$ and, hence, $r'[i][Y']$ and $r'[j][Y']$ are again the same tuples (due to the construction of $r'$). 
\end{proof}
The proof (in the Appendix) selects between a reduction from MPD with the singleton $\set{\tA\ra\tB}$ and a reduction from MPD with the matching constraint $\tA\lra\tB$, depending on the structure of the $F$.

We note that the assumption that $F$ has no consensus FDs is necessary. For example, for $F=\set{\emptyset\ra \tA}$, which is nontrivial, we can find a most probable database by considering every possible value $a$ for $\tA$, computing the probability of selecting $a$ in all distributions, and finally using the value with the maximal probability. 

From Theorem~\ref{thm:all-uncertain-all-hard} we immediately conclude the hardness of the three problems on \e{every} nontrivial set of FDs in the block-independent-disjoint (BID) model of probabilistic databases~\cite{ReS07}, due to the translation mentioned in the Introduction. 
\section{General Sets of Unary Functional Dependencies}\label{sec:unary}
In Section~\ref{sec:binary}, we studied the complexity of the three problems in the case of a binary schema, and we gave a full classification of the different possible sets of FDs. In this section, we extend these results to a general classification (dichotomy) for every set of \e{unary} FDs, that is, FDs with a single attribute on the left side. Our result uses a decomposition technique that we devise next.

\subsection{Reduction by Decomposition}
In this section, we devise a decomposition technique that allows us to reduce our computational problems from one set of FDs into multiple smaller subsets of the set. This technique is stated in the next theorem. After the theorem, we show several consequences that illustrate the use of the technique. Later, we will use these consequences to establish a full classification of complexity for the sets of unary FDs.

\begin{rtheoremrep}\label{thm:reduce-decompose}
Let $F$ be a set of FDs over a relation schema $R$. Suppose that $F=F_1\cup F_2$ and that all attributes in $\atts(F_1)\cap\atts(F_2)$ are certain (unmarked). 
Each of the three problems (in Definition~\ref{def:problems}) can be solved in polynomial time if its version with $F_j$ and $\atts(F_j)$ is solvable in polynomial time for both $j=1$ and $j=2$.
\end{rtheoremrep}
\begin{proofsketch}
Let $U_j=\pi_{\atts(F_j)}U$ for $j=1,2$. We show the following:
\begin{enumerate}
    \item $U$ is possibly consistent w.r.t.~$F$ if and only if $U_1$ and $U_2$ are possibly consistent w.r.t.~$F_1$ and $F_2$, respectively. 
    \item MPDs of $U_1$ and $U_2$ can be easily combined to produce an MPD of $U$.
    \item $\prob_U(F) = \prob_{U_1}(F_1) \cdot \prob_{U_2}(F_2)$.
\end{enumerate}
The full details are provided in the Appendix.
\end{proofsketch}

\begin{proof}
We start with the problem of possible consistency. Let $U_j=\pi_{\atts(F_j)}U$ for $j=1,2$, and let $U'=\pi_{\atts(U)\setminus\atts(F)}U$. We show that $U$ is possibly consistent w.r.t.~$F$ if and only if $U_j$ is possibly consistent w.r.t.~$F_j$ for $j=1,2$. If $r$ is a  consistent sample of $U$ then $\pi_{\atts(F_1)}r$ is a  consistent sample of $U_1$ w.r.t.~$F_1$ and $\pi_{\atts(F_2)}r$ is a consistent sample of $U_2$ w.r.t.~$F_2$.

Next, assume that $r_1$ is a consistent sample of $U_1$  w.r.t.~$F_1$ and $r_2$ is a consistent sample of $U_2$  w.r.t.~$F_2$. Let $r'$ be a sample in the support of $U'$. We show that $r_1+r_2+r'$ is a consistent sample of $U$ w.r.t.~$F$, where $r_1+r_2+r'$ is the relation $r$ such that $\tids(r)=\tids(r_1)$, that $\atts(r)=\atts(r_1)\cup\atts(r_2)\cup\atts(r')$, and that 
$r[i]=r_1[i] \cup r_2[i] \cup r_3[i]$ (i.e., the natural join of the tuples
$r_1[i], r_2[i]$ and $r'[i]$) for all $i\in\tids(r)$. Observe that
$$\tids(r_1)=\tids(r_2)=\tids(r')=\tids(U_1)=\tids(U_2)=\tids(U')=\tids(U) \,,$$
and that for every tuple identifier $i$ it is the case that the three relations ($r_1, r_2$ and $r'$) agree on the common attributes of tuple $i$, since these are certain attributes. Therefore, $r_1, r_2$ and $r'$ can be naturally combined to produce a relation. This relation is consistent since every FD is covered by one of the $F_j$s. Moreover, the addition of $r'$ (that has a disjoint set of attributes) does not change the consistency.

For the problem of most probable database, if $r_1$ and $r_2$ are most probable databases of $U_1$ and $U_2$, respectively, then by similar arguments, $r_1 + r_2 + r'$ is a most probable database of $U$ for every sample $r'$ in the support of $U'$. 

Finally, for the probability of consistency, we show that $\prob_U(F) = \prob_{U_1}(F_1) \cdot \prob_{U_2}(F_2)$. Let $r$ denote the random element that corresponds to a sample of $U$,
let $r_1$ be the random element $\pi_{\atts(F_1)}r$, and 
$r_2$ be the random element $\pi_{\atts(F_2)}r$. Then $r\models F$ if and only if $r_1 \models F_1$ and $r_2 \models F_2$. Moreover, $r_1$ and $r_2$ are probabilistically independent, since they involve disjoint sets of distributions, and in particular, their consistencies are probabilistically independent. This completes the proof.
\end{proof}

An immediate conclusion from Theorem~\ref{thm:reduce-decompose} is that we can eliminate the FDs that involve only certain attributes if we know how to deal with the remaining FDs.
\begin{corollary}
Let $F$ be a set of FDs over a relation schema $R$. Let $X\ra Y$ be an FD in $F$, and suppose that all attributes in $X$ and $Y$ are certain. Then each of the three problems (in Definition~\ref{def:problems}) is polynomial-time reducible to its version with $R$ and $F\setminus\set{X\ra Y}$.
\end{corollary}

\begin{remark} Eliminating the FDs over the certain attributes is not always beneficial, since these FDs might be needed for applying a polynomial-time algorithm. As an example, consider the following set of FDs: $\set{\uca A\ra B \,,\, B\ra C\,,\, C\ra\uca{A}}$.
As we will show later, for this set of FDs we can find a most probable database in polynomial time. However, we will also show that possible consistency is NP-hard for the subset
$\set{\uca A\ra B \,,\, C\ra\uca{A}}$. Hence, $B\ra C$ is needed for the polynomial-time algorithm.\qed
\end{remark}

The following consequence of Theorem~\ref{thm:reduce-decompose} identifies a general tractable case: the problems are solvable in polynomial time if uncertain attributes do not appear in the left side of the FDs (but they can appear in the right side or outside of the FDs). 
\begin{rtheoremrep}\label{thm:lhs-certain-poly}
Let $F$ be a set of FDs. If the left side of every FD includes only certain attributes, then each of the three problems (in Definition~\ref{def:problems}) is solvable in polynomial time.
\end{rtheoremrep}
\begin{proofsketch}
Assume, without loss of generality, that each FD in $F$ contains a single attribute on the right side. For every $A\in\atts(F)$, let $F_A$ be the subset of $F$ that contains all FDs with $A$ being the right side (i.e., all FDs of the form $X\ra A$). Then 
$F=\cup_{A\in\atts(F)}F_A$. Note that sets $F_A$ and $F_B$, where $A\neq B$, share only certain attributes. This is true since our assumption implies that an uncertain attribute $\uca A$ can appear only in $F_{\uca A}$. Hence, we can apply Theorem~\ref{thm:reduce-decompose} repeatedly and conclude that we need a polynomial-time solution for each $F_{\uca A}$. In the Appendix, we show that we can obtain that using a similar concept to the algorithm for $A\ra\uca B$ from Section~\ref{sec:algorithms-binary}.
\end{proofsketch}

\begin{proof}
Following the proof sketch, we show that we can solve each of the three problems for each $F_{\uca A}$ desperately. Consider the undirected graph $G$ that has as vertices the tuple identifiers of $U$, and an edge between two identifiers $i$ and $i'$ whenever the two agree on all attributes of the left hand side of some dependency in $F_{\uca A}$. Hence, an edge means that the two tuples should have the same value in every consistent sample. Therefore, we can compute the most probable database, and calculate the probability of consistency, in a similar manner to the algorithms for $A\ra\uca B$ from Section~\ref{sec:algorithms-binary}, except that now we consider entire connected components rather than simple groups by the attribute $A$.
\end{proof}

\subsection{Classification}
We now state the precise classification of the complexity of the problems in the case of unary FDs. The statement uses the following terminology. Let $F$ be a set of unary FDs. Recall that two attributes $A$ and $B$ and are \e{equivalent} if they have the same closure, that is, $\closure{\set{A}}{F}=\closure{\set{B}}{F}$. An attribute $A$ is called a \e{sink} if 
$\closure{\set{A}}{F}=\set{A}$, that is, $A$ does not appear in the left hand side of any nontrivial FD.
In this section, we will prove the following classification (trichotomy) result, which is also illustrated in Figure~\ref{ref:fig:unary-classification}.
\begin{theorem}\label{thm:unary:trichotomy}
Let $F$ be a set of unary FDs over a relation schema $R$. Then following hold. 
\begin{enumerate}
    \item If every uncertain attribute is either a sink or equivalent to a certain attribute, then a most probable database can be found in polynomial time; otherwise, possible consistency is NP-complete.
    \item If every uncertain attribute is a sink, then the probability of consistency can be calculated in polynomial time; otherwise, it is $\fpsharpp$-complete.
\end{enumerate}
\end{theorem}
The following examples illustrate the instantiation of the theorem to specific scenarios.

\begin{example}
We give several examples for the case of a ternary schema $\set{A,B,C}$. Consider the following sets of FDs:
$$ 
F_1\defeq\set{A\ra B\ra \tC} \quad\quad
F_2\defeq\set{A\ra \tB\ra C} \quad\quad
F_3\defeq\set{A\lra \tB\ra \tC}
$$
Theorem~\ref{thm:unary:trichotomy} tells us that the following.
All three problems are solvable in polynomial time in the case of $F_1$, since the uncertain attribute $\tC$ is a sink. In the case of $F_2$, we can see that $\tB$ is neither a sink nor equivalent to any certain attribute; hence, all three problems are intractable for $F_2$. In the case of $F_3$, it holds that $\tC$ is a sink and $\tB$ is not a sink but is equivalent to the certain attribute $A$. Hence, the consistency of $F_3$ is $\fpsharpp$-complete, but we can find a most probable database in polynomial time.\qedexample
\end{example}

\begin{example}
Consider again the CIR $U_2$ of Figure~\ref{fig:spokesperson-cir}. Consider the following constraints.
\begin{enumerate}
\item $\att{business}\ra\uca{\att{spokesperson}}\uca{\att{location}}$
\item $\uca{\att{spokesperson}}\ra\uca{\att{location}}$
\item $\att{business}\lra\uca{\att{spokesperson}}\ra\uca{\att{location}}$
\end{enumerate}
For the first constraint, all three problems are tractable since both $\uca{\att{spokesperson}}$ and $\uca{\att{location}}$ are sinks. For the second constraint, all three problems are intractable since 
$\uca{\att{spokesperson}}$ is neither a sink nor equivalent to any certain attribute. For the third constraint, a most probable database can be found in polynomial time since $\uca{\att{spokesperson}}$  is equivalent to the certain $\att{business}$ and $\uca{\att{location}}$ is a sink, but the probability of consistency is $\fpsharpp$-complete since $\uca{\att{spokesperson}}$ is not a sink.
\qedexample
\end{example}


In the remainder of this section, we prove each of the two parts of Theorem~\ref{thm:unary:trichotomy} separately.

\begin{figure}
  \centering
  \includegraphics[width=1.0\linewidth]{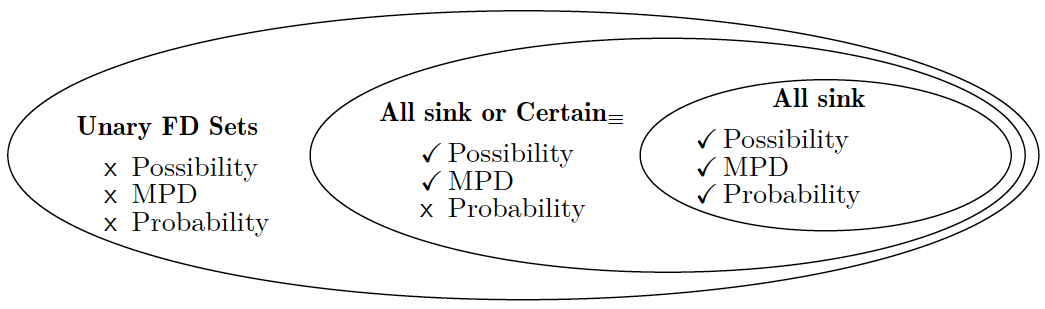}
  \caption{Classification of the complexity of consistency problems for sets of unary FDs. (See Table~\ref{table:complexity:binary} for the naming of the problems.) ``All sink'' refers to the case where every uncertain attribute is sink, and ``All sink or Certain$_\equiv$'' refers to the case where every uncertain attribute is either a sink or equivalent to a certain attribute. \label{ref:fig:unary-classification}}
  \end{figure}


\subsubsection
{Part 1 of Theorem~\ref{thm:unary:trichotomy} (Possible Consistency and MPD)}

We first prove the tractability side of Part~1 of the theorem.
\begin{rlemmarep}\label{lemma:unary-tractable}
Let $F$ be a set of unary FDs over a schema $R$. If every uncertain attribute is either a sink or equivalent to a certain attribute, then a most probable database can be found in polynomial time.
\end{rlemmarep}
\begin{proofsketch}
The idea to define a set $F_{\uca A}$ of FDs for every uncertain attribute $\uca A\in\uatts(U)$, and a set $F'$ of FDs where all left-side attributes are certain, such that $F$ is equivalent to
$ F'\cup \bigcup_{\uca A\in\uatts(U)}F_{\uca A} $.
Then, we repeatedly apply Theorem~\ref{thm:reduce-decompose} to reduce the original problem to instances that are solvable in polynomial time by Proposition~\ref{prop:A_lra_tildeB_poly} and Theorem~\ref{thm:lhs-certain-poly}.
\end{proofsketch}

\begin{proof}
Note that a sink uncertain attribute $\uca A$ cannot appear in the left side of any FD in $F$, unless it is the trivial FD $\uca A\ra \uca A$. We assume that $F$ does not contain such trivial FDs, and so, a sink attribute can appear only on the right side of a rule.
For every uncertain attribute $\uca A\in\uatts(U)$ we define a set $F_{\uca A}$ of FDs as follows.
\begin{itemize}
\item If $\uca A$ is a sink or not at all in $\atts(F)$, then $F_{\uca A}=\emptyset$.
\item If $\uca A$ is non-sink, then we find an equivalent certain attribute $B$ and set $F_{\uca A}=\set{\uca A\lra B}$.
\end{itemize}
Then $F$ is equivalent to the set 
\begin{equation}\label{eq:dichotomy-tractable-union}
F'\cup \bigcup_{\uca A\in\uatts(U)}F_{\uca A}
\end{equation} where $F'$
is obtained from $F$ by replacing every non-sink uncertain attribute $\uca A$ by its partner $B$ in $F_{\uca A}$. Now, observe the following.
\begin{itemize}
    \item For each $F_{\uca A}$ we can find a most probable database in polynomial time, according to Proposition~\ref{prop:A_lra_tildeB_poly}.
    \item In $F'$ we have that all left-side attributes are certain, and so, we can find a most probable database according to Theorem~\ref{thm:lhs-certain-poly}.
\end{itemize}
Finally, observe that every uncertain attribute appears in at most one subset in Equation~\eqref{eq:dichotomy-tractable-union}. Hence, we can repeatedly apply Theorem~\ref{thm:reduce-decompose} and complete the proof.
\end{proof}

For the hardness side of Part~1 of Theorem~\ref{thm:unary:trichotomy}, we will need the following lemma, which generalizes the case of $\uca{A}\lra\uca{B}$ from Proposition~\ref{prop:hardness-tildeA-ra-B-and-tildaA-lra-tildeB}.

\begin{rlemmarep}\label{lemma:equiv-uncertain-hard}
Let $R=\set{\uca{A_1},\dots,\uca{A_k}}$ consist of $k>1$ uncertain attributes, and suppose that $F$ is a set of FDs stating that all attributes in $R$ are equivalent. Then possible consistency is NP-complete.
\end{rlemmarep}
\begin{proof}
The case of $k=2$, namely $\uca{A_1}\lra\uca{A_2}$, has been shown in Lemma~\ref{lemma:hardness-tildeA-lra-tildeB}. For $k>2$, we will show a reduction from $\uca {A_1}\lra\uca{A_2}$. Given an input $U_0$ for $\set{\uca {A_1},\uca{A_2}}$, we construct an instance $U$ for $\set{\uca {A_1},\dots,\uca{A_k}}$ by adding to $U_0$ columns to the right. We need 
$k-2$ such columns. We set in each cell of these $k-2$ columns the same distribution: the uniform distribution over $m$ distinct values, where $m=|\tids(U_0)|$ is the number of tuples in $U_0$.

The correctness of the reduction is due to the following observations.
First, if $r$ is a  consistent sample for $U$, then $\pi_{\set{\uca{A_1},\uca{A_2}}}r$ is a consistent sample for $U_0$. Second, every consistent sample $r_0$ of $U_0$ can be extended into a feasible consistent sample of $U$ by adding to each tuple a value that is uniquely determined by the values in the row. Third, the probability of every sample $r$ of $U$ is determined only by the values in $\uca{A_1}$ and $\uca{A_2}$:
\[\prob_{U}(r) = \prob_{U_0}(\pi_{\set{\uca{A_1},\uca{A_2}}}r)\cdot m^{-m(k-2)}\]
In particular, finding a consistent sample $r$ for $U$ is the same problem as finding a consistent sample $r_0$ for $U_0$.
\end{proof}

The next lemma states the hardness side of Part~1 of Theorem~\ref{thm:unary:trichotomy}.

\begin{rlemmarep}\label{lemma:unary-intractable}
Let $F$ be a set of unary FDs over a schema $R$. If there is an uncertain attribute that is neither a sink nor equivalent to a certain attribute, then possible consistency is NP-complete.
\end{rlemmarep}
\begin{proofsketch}
Let $\tA$ be an attribute that is neither a sink nor equivalent to a certain attribute. Let $X$ be the
closure of $\tA$ and $X'$ be $X\setminus\set{\tA}$. Observe the following. First, $X'$ must be nonempty since $\tA$ is not a sink. Second, if any attribute in $X'$ implies $\tA$ then it is equivalent to $\tA$, and then it is necessarily uncertain. 
We consider two cases:
\begin{enumerate}
    \item No attribute in $X'$ implies $\tA$.
    \item Some attribute in $X'$ implies $\tA$.
\end{enumerate}
For the first case, we show a reduction from $\tA\ra B$, where possible consistency is NP-complete due to Proposition~\ref{prop:hardness-tildeA-ra-B-and-tildaA-lra-tildeB}. For the second case, let $\tB_1,\dots,\tB_\ell$ be the set of all attributes in $X'$ that imply $\tA$. As said above, each $\tB_j$ must be uncertain. Then all of 
$\tB_1,\dots,\tB_\ell,\tA$ are equivalent. We show a reduction from the problem of
Lemma~\ref{lemma:equiv-uncertain-hard} where $k=\ell+1$. The constructions of the reductions for the two cases and the proofs of correctness can be found in the Appendix.
\end{proofsketch}

\begin{proof}
We consider the two cases that are specified in the proof sketch of Lemma~\ref{lemma:unary-intractable} (in the main body of the paper).

\subparagraph*{Case 1: No attribute in $X'$ implies $\tA$.} We show a reduction from $\tA\ra B$, where possible consistency is NP-complete due to Proposition~\ref{prop:hardness-tildeA-ra-B-and-tildaA-lra-tildeB}. Let $U_0$ be an input for $\set{\tA,B}$ and $\tA\ra B$. We construct an input $U$ for $R$ and $F$ as follows. The two CIRs have the same set of identifiers, that is, $\tids(U)=\tids(U_0)$. For each $i\in\tids(U)$ we define the tuple $U[i]$ by:
\begin{itemize}
    \item $U[i][\tA]=U_0[i][\tA]$;
    \item $U[i][B]=U_0[i][B]$ for every $B\in X'$ (even if $B$ is uncertain);
    \item $U[i][C]=i$ for every $C\in R\setminus X$.
\end{itemize}
Note that we are using the tuple identifier $i$ as a value to assure that every two tuples have different values in the corresponding $C$ attributes. We will prove that $U$ is possibly consistent if and only if $U_0$ is possibly consistent. From a  consistent sample $r$ for $U$ we can get a  consistent sample $r_0$ for $U_0$ by projecting on $\tA$ and some attribute $C\in X$, and renaming $C$ as $B$. Note that $r_0$ is a  consistent sample of $U_0$ since $r$
satisfies $\tA\ra C$. 

For the other direction, suppose that $r_0$ is a  consistent sample for $U_0$. To obtain a  consistent sample for $U$, we select each $\tA$ value $U[i][\tA]$ to be one chosen by $r_0$, namely
$r_0[i][\tA]$. Denote the result by $r$. We need to show that $r$ satisfies $F$. Let $D\ra E$ an FD in $F$. If $D=\tA$ then $E\in X$ and then the FD is satisfied since $r_0$ satisfies $\tA\ra B$. If $D\in X'$ then $E\in X$, and $E\neq\tA$ since $D$ and $\tA$ are not equivalent. Hence, $E\in X'$ as well. It thus follows that the $E$ column is equal to the $D$ column and, therefore, $D\ra E$ is satisfied. Finally, if $D\notin X$ then no two tuples agree on $D$, and then $D\ra E$ is satisfied in a vacuous manner.

\subparagraph*{Case 2: Some attribute in $X'$ implies $\tA$.} Let $\tB_1,\dots,\tB_\ell$ be the set of all attributes in $X'$ that imply $\tA$. As said above, each $\tB_j$ must be uncertain. Then all of 
$\tB_1,\dots,\tB_\ell,\tA$ are equivalent. We will show a reduction from the problem of
Lemma~\ref{lemma:equiv-uncertain-hard} where $k=\ell+1$. Let $U_0$ be an input for that problem; hence, $U_0$ is a CIR over $\set{\tA_1,\dots,\tA_k}$. We construct a CIR $U$ over $R$, as follows. Again, $\tids(U)=\tids(U_0)$ and for each $i\in\tids(U)$ we define the tuple $U[i]$ according to the following rules. We use an arbitrary constant value that we denote by $c$.
\begin{itemize}
    \item $U[i][\tA]=U_0[i][\tA_k]$;
    \item $U[i][\tB_j]=U_0[i][\tA_j]$ for $j=1,\dots,\ell$;
    \item $U[i][C]=c$ for every $C\in X'\setminus\set{\tB_1,\dots,\tB_\ell}$; 
    \item $U[i][C]=i$ for every $C\in R\setminus X$.
\end{itemize}
Again, we show that $U$ is possibly consistent if and only if $U_0$ is possibly consistent. From a  consistent sample $r$ for $U$ we can get a consistent sample $r_0$ for $U_0$ by projecting on $\set{\tB_1,\dots,\tB_\ell,\tA}$, renaming each $\tB_j$ as $\tA_j$ and $\tA$ as $\tA_k$. Note that $r_0$ is a  consistent sample for $U_0$ (and in particular $r_0$ satisfies the FDs of Lemma~\ref{lemma:equiv-uncertain-hard}) since $r$ satisfies $\tB_1\lra\dots\lra\tB_\ell\lra\tA$.

For the other direction, let $r_0$ be a  consistent sample of $U_0$. For each tuple identifier $i\in\tids(U)$ we need to state the choice for the value $U[i][\tA]$ and $U[i][\tB_j]$ for $j=1,\dots,\ell$. We set the former to
 $r_0[i][\tA_k]$ and the latter to $r_0[i][\tA_j]$. Let the result be $r$. We need to show that $r$ satisfies $F$. Let $D\ra E$ be an FD in $F$. We have the following cases:
 \begin{itemize}
      \item $D\notin X$. Then no two tuples agree on $D$, and then $D\ra E$ is satisfied in a vacuous manner.
      \item $E\notin X$. Then $D\notin X$ (previous case) since $X$ is a closure.
       \item $E\in X\setminus\set{\tA,\tB_1,\dots,\tB_\ell}$. Then $D\ra E$ holds since all values of the attribute $E$ are equal (to $c$) in $r$.
       \item $D\in X\setminus\set{\tA,\tB_1,\dots,\tB_\ell}$. Then $E\in X$ and $E$ cannot be in
       $\set{\tA,\tB_1,\dots,\tB_\ell}$, or otherwise $D$ is also equivalent to $\tA$ (while we know that only $\tB_1,\dots,\tB_\ell$ are equivalent to $\tA$).
      \item $\set{D,E}\subseteq\set{\tA,\tB_1,\dots,\tB_\ell}$. Then $D\ra E$ is satisfied since $r_0$ satisfies the constraint $\tA_1\lra\dots\lra\tA_k$.
       \end{itemize}
 We conclude that $r$ satisfies $F$, and so, $r$ is a  consistent sample. This concludes the proof.
\end{proof}

\subsubsection
{Part 2 Theorem~\ref{thm:unary:trichotomy} (Probability of Consistency)}
We now move on to Part~2. The tractability side follows immediately from Theorem~\ref{thm:lhs-certain-poly}, since if all uncertain attributes are sinks, then all left-side attributes are certain (up to trivial FDs $\tA\ra\tA$ that can be ignored).
Hence, it remains to prove the hardness side of Part~2 of Theorem~\ref{thm:unary:trichotomy}. We start with the following lemma, where we use a reduction from the case of $A\lra\uca B$, where probability of consistency is $\fpsharpp$-complete by Proposition~\ref{prop:hardness-A-lra-tildeB}, to establish hardness for a more general case.

\begin{rlemmarep}\label{lemma:equiv-prob-sharpp}
Let $F$ be a set of unary FDs over a schema $R$. If at least one uncertain attribute is equivalent to a certain attribute, then the probability of consistency is 
$\fpsharpp$-complete.
\end{rlemmarep}
\begin{proof}
Let $A$ be a certain attribute, and let $\uca B$ be an uncertain attribute that is equivalent to $A$. We will use Proposition~\ref{prop:hardness-A-lra-tildeB} and show a reduction from 
$R_0=\set{A,\uca B}$ and $F_0=\set{A\lra\uca B}$. Let $U_0$ be an input for
$R_0$ and $U_0$. We will produce an input $U$ for $R$ and $F$ by adding columns with certain values, regardless of whether the attribute is certain or not, and we will show that the probability of consistency if the same in $U$ and $U_0$. The construction is straightforward: we simply copy the $A$ column. In notation, for each identifier $i\in\tids(U)$ and attribute $C\in R\setminus\set{A,B}$ we set $U[i][C]=U_0[i][A]$.

Observe that every sample $r_0$ of $U_0$ has a unique extension $r_0^+$ of $U$, and every sample of $U$ is the unique extension $r_0^+$ of some sample $r_0$ of $U$. To complete the proof, we show that both of the following hold:
\begin{enumerate}
\item $\prob_{U_0}(r_0)=\prob_{U}(r_0^+)$;
\item $r_0\models F_0$ if and only if $r_0^+\models F$.
\end{enumerate}
Part~1 is straightforward from the construction, since all values of $U$ are certain except for those in $\uca B$. So we now prove Part~2. The ``if'' direction is due to the fact that $F$ implies both
$A\ra \uca B$ and $\uca B\ra A$. For the ``only if'' direction, suppose that $r_0\models F_0$ and let $r=r_0^+$. We show that $r\models F$. Let $C\ra D$ be an FD in $F$, let $i$ and $j$ be two identifiers, and suppose that $r[i][C]=r[j][C]$. We need to show that $r[i][D]=r[j][D]$. Due to our construction, $r[i][C]=r[j][C]$ implies that either $r[i][A]=r[j][A]$ or $r[i][\uca B]=r[j][\uca B]$ (in case $C$ is $\uca B$). Since $r$ satisfies $A\ra \uca B$, it holds that
both $r[i][A]=r[j][A]$ and $r[i][\uca B]=r[j][\uca B]$ hold. Finally, observe that $D$ is either $\uca B$ or a copy of $A$ (from our construction). Therefore, $r[i][D]=r[j][D]$, as claimed.
\end{proof}

We can now complete the proof of the hardness side of Part~2.
\begin{rlemmarep}\label{lemma:nonsink-prob-sharpp}
Let $F$ be a set of unary FDs over a schema $R$. If there is at least one uncertain attribute that is not a sink, then the probability of consistency is $\fpsharpp$-complete.
\end{rlemmarep}
\begin{proofsketch}
Let $\tA$ be an uncertain attribute that is not a sink. Let $Y=(\tA)^+_F\setminus\set{\tA}$. Note that $Y$ is nonempty, since $\tA$ is not a sink. If any attribute $B$ in $Y$ functionally determines $\tA$, then we can use this attribute as a certain attribute (even if it is uncertain) and use Lemma~\ref{lemma:equiv-prob-sharpp}, since $\tA$ is equivalent to $B$. Otherwise, suppose that no attribute in $Y$ determines $\tA$. For this case, we show (in the Appendix) a reduction from $\tA\ra B$, where the probability of consistency is $\fpsharpp$-hard according to Proposition~\ref{prop:hardness-tildeA-ra-B-and-tildaA-lra-tildeB}.
\end{proofsketch}

\begin{proof}
Following the proof sketch, we assume that no attribute in $Y$ determines $\tA$. For this case, we show a reduction from $\tA\ra B$, where the probability of consistency is $\fpsharpp$-hard according to Proposition~\ref{prop:hardness-tildeA-ra-B-and-tildaA-lra-tildeB}.

Let $R_0=\set{\tA,B}$ and $F_0=\set{\tA\ra B}$. Let $U_0$ be an input for
$R_0$ and $U_0$. We will again produce an input $U$ for $R$ and $F$ by adding columns with certain values, regardless of whether the attribute is certain or not, and we will show that the probability of consistency is the same in $U$ and $U_0$. The construction is as follows: we copy the $B$ column to every attribute in $Y$, and use a unique fresh value for every other attribute. In notation, for each identifier $i\in\tids(U)$ and attribute $C\in Y$ we set $U[i][C]=U_0[i][B]$, and for every attribute $C\in R\setminus(Y\cup\set{\tA})$ we set $U[i][C]=i$.

We show the correctness of the reduction, using a similar argument as in the proof of Lemma~\ref{lemma:equiv-prob-sharpp}. Note that our construction is such that every sample $r_0$ of $U_0$ has a unique extension $r_0^+$ of $U$,
and every sample of $U$ is the unique extension $r_0^+$ of some sample $r_0$ of $U$. Then, the following hold:
\begin{enumerate}
\item $\prob_{U_0}(r_0)=\prob_{U}(r_0^+)$.
\item $r_0\models F_0$ if and only if $r_0^+\models F$;
\end{enumerate}
Again, Part~1 is straightforward from the construction, so we prove Part~2. The ``if'' direction holds since $F\models\tA\ra B$. For the ``only if'' direction, suppose that $r_0\models F_0$ and let $r=r_0^+$. We show that $r\models F$. Let $C\ra D$ be an FD in $F$, let $i$ and $j$ be two different identifiers, and suppose that $r[i][C]=r[j][C]$. We need to show that $r[i][D]=r[j][D]$.
Note that the construction implies that $C$ must be in the closure of $\tA$, since $i\neq j$.
If $C$ is $\tA$, then $D$ belongs to $Y$, and then  $r[i][D]=r[j][D]$ since $r_0$ is consistent. 
If $C$ belongs to $Y$, then $D$ also belongs to $Y$ (since $D$ cannot be $\tA$), and then  $r[i][D]=r[j][D]$ by construction. This completes the proof.
\end{proof}

\subsubsection
{Recap}

We can now complete the proof of Theorem~\ref{thm:unary:trichotomy}. For Part~1, the tractability side is given by Lemma~\ref{lemma:unary-tractable}, and the hardness is given by Lemma~\ref{lemma:unary-intractable}. As for Part~2, the tractability side follows immediately from Theorem~\ref{thm:lhs-certain-poly}, and the hardness side is stated in Lemma~\ref{lemma:nonsink-prob-sharpp}.
\section{Conclusions}\label{sec:conclusions}
We defined the concept of a CIR and studied the complexity of three problems that relate to consistency under FDs: possible consistency, finding a most probable database, and the probability of consistency. A seemingly minor feature of the definition of a CIR is the distinction between certain and uncertain attributes; yet, this distinction turns out to be crucial for detecting tractable cases. We gave  classification results for several classes of FD sets, including a single FD, a matching constraint, and arbitrary sets of unary FDs. We also showed that if all attributes are allowed to be uncertain, then the first two problems are intractable for every nontrivial set of FDs. 

This work leaves many problems for future investigation. Within the model, we have not yet completed the classification for the whole class of FD sets, where the problem remains open. Recall that a full classification is known for the most probable database for tuple-independent databases~\cite{DBLP:journals/tods/LivshitsKR20}. Moreover, as we hit hardness already for simple cases (e.g., $\tA\ra B$), it is important to identify realistic properties of the CIR that reduce the complexity of the problems and allow for efficient algorithms.

Going beyond the framework of this paper, we plan to study additional types of constraints that are relevant to data cleaning~\cite{DBLP:series/synthesis/2012Fan}, such as conditional FDs, denial constraints, and foreign-key constraints (where significant progress has been recently made in the problem of consistent query answering~\cite{DBLP:conf/pods/HannulaW22}). Another useful direction is to consider \e{soft} or \e{approximate} versions of the constraints, where it suffices to be consistent to some quantitative extent~\cite{DBLP:conf/icdt/CarmeliGKLT21,DBLP:conf/pods/JhaRS08,DBLP:journals/pvldb/LivshitsHIK20}. Finally, we have made the assumption of probabilistic independence among the cells as this is the most basic setting to initiate this research. To capture realistic correlations in the database noise, it is important to extend this work to data models that allow for (learnable) probabilistic dependencies, such as Markov Logic~\cite{DBLP:journals/pvldb/JhaS12}.

\newpage

\bibliography{main}

\end{document}